\documentclass{article}
\usepackage{graphicx}
\usepackage{subcaption}
\usepackage{caption}

\usepackage{arxiv}

\newcommand{\update}[1]{#1}
\newcommand{\updateNew}[1]{#1}

\usepackage[utf8]{inputenc} 
\usepackage[T1]{fontenc}    
\usepackage{url}            
\usepackage{booktabs}       
\usepackage{amsfonts}       
\usepackage{nicefrac}       
\usepackage{microtype}      
\usepackage{lipsum}		
\usepackage{graphicx}
\usepackage{natbib}
\usepackage{doi}
\usepackage{multirow}
\usepackage{amssymb,amsmath,amsthm}
\usepackage{enumitem}
\usepackage[nameinlink,capitalize,noabbrev]{cleveref}
\usepackage{xcolor}
\usepackage{mathtools}
\usepackage{algorithm}
\usepackage{siunitx}
\usepackage{algpseudocode}
\usepackage{hyperref}       
\hypersetup{
    colorlinks,
    linkcolor={red!60!black},
    citecolor={blue!60!black},
    urlcolor={blue!60!black}
}

\usepackage{bbm}
\usepackage{graphicx}
\usepackage{subcaption}
\usepackage{caption}
\crefname{step}{step}{steps}
\Crefname{step}{Step}{Steps}
\newtheorem{theorem}{Theorem}
\newtheorem{remark}{Remark}

\newtheorem{definition}{Definition}
\newtheorem{lemma}{Lemma}
\newtheorem{assumption}{Assumption}

\DeclareMathOperator{\argmax}{arg \ max}

\title{Data-driven Koopman MPC using Mixed Stochastic-Deterministic Tubes}

\author{Zhengang Zhong\\
	Sargent Centre for Process Systems Engineering\\
	Imperial College London\\
    London, England\\
	\texttt{z.zhong20@imperial.ac.uk} \\
	\And
	Ehecatl Antonio del Rio-Chanona\\
	Sargent Centre for Process Systems Engineering\\
	Imperial College London\\
	London, England\\
	\texttt{a.del-rio-chanona@imperial.ac.uk} \\
	\And
	Panagiotis Petsagkourakis\\
	Sargent Centre for Process Systems Engineering\\
	Imperial College London\\
	London, England\\
	\texttt{p.petsagkourakis@imperial.ac.uk} \\
}

\renewcommand{\shorttitle}{Gradient-based Training for Multilevel Optimal Control Problems}

\hypersetup{
pdftitle={multilevel Optimal Control Problems using Gradient-based Training for Neural Surrogates},
pdfsubject={OC},
pdfauthor={Dante~Kalise, Estefania~Loayza-Romero, Kirsten A.~Morris, Zhengang~Zhong},
pdfkeywords={},
}

\begin{document}
\maketitle

\begin{abstract}
This paper presents a novel data-driven stochastic MPC design for discrete-time nonlinear systems with additive disturbances by leveraging the Koopman operator and a distributionally robust optimization (DRO) framework. By lifting the dynamical system into a linear space, we achieve a finite-dimensional approximation of the Koopman operator. We explicitly account for the modeling approximation and additive disturbance error by a mixed stochastic-deterministic tube for the lifted linear model. This ensures the regulation of the original nonlinear system while complying with the prespecified constraints. Stochastic and deterministic tubes are constructed using a DRO and a hyper-cube hull, respectively. We provide finite sample error bounds for both types of tubes. The effectiveness of the proposed approach is demonstrated through numerical simulations.
\end{abstract}

\keywords{Koopman operator \and Stochastic model predictive control}

\section{INTRODUCTION}
Model predictive control (MPC) is a control scheme that repeatedly solves optimization problems online based on a system model and prescribed constraints to decide optimal control actions \cite{mayne2000constrained, grune2017nonlinear}. 
When the system model is not acquirable through first principles, leveraging data-driven methods is a reliable alternative.
 Koopman operator theory enables obtaining a global linear representation of nonlinear dynamical systems from data \cite{kaiser2021data, mezic2005spectral}, which has attracted a lot of attention due to its benign properties. The Koopman operator acting on an observable function describes the evolution of this observable function along the trajectory of the original system in both deterministic and stochastic scenarios \cite{budivsic2012applied, wanner2022robust, vcrnjaric2020koopman}. Being a linear operator the Koopman operator has eigenvalues, eigenfunctions, and modes, which allow a low-dimensional interpretation of the original high-dimensional system \cite{brunton2016koopman}. Methods practically concerning  with a suitable finite-dimensional approximation for the Koopman
operator have been proposed to learn the leading Koopman eigenvalues, eigenfunctions, and modes directly from data, such as Dynamic Mode Decomposition (DMD) \cite{schmid2010dynamic}, Hankel-DMD \cite{arbabi2017ergodic}, noise-resistant DMD \cite{wanner2022robust} and Extended DMD (EDMD) \cite{folkestad2020extended, williams2015data, williams2014kernel}.
Other works include  \cite{korda2020optimal} introduces an approach to learn the invariant subspaces of the Koopman operator spanned by generalized eigenfunc-functions, and more recently,  \cite{khosravi2022representer} and \cite{kostic2022learning}  learn an invariant subspace of the corresponding Koopman operator by leveraging reproducing kernel.


The methods mentioned above are mainly applied to learn the Koopman operator corresponding to autonomous dynamical systems. However, for controlled dynamical systems, learning the corresponding Koopman operator is more involved. The Koopman operator with control inputs \cite{kaiser2021data} can be interpreted as an operator acting on the space observables,
in which each observable is an element dependent on 
the state and/or inputs \cite{proctor2018generalizing}. The paper \cite{korda2018linear} defines the Koopman operator associated with controlled systems as an operator evolving observables of the uncontrolled dynamical system on the extended state space. For such controlled nonlinear systems, the approximated Koopman operator serves as a linear predictor for the design of a linear MPC on the space of observables. For control affine systems, the corresponding lifted systems can be formulated as a bilinear control system in the Koopman space based on the Koopman Canonical Transform \cite{surana2016koopman}.
By regarding control inputs as a parameter of the Koopman operator, a family of Koopman operators corresponding to each input can be defined for the nonlinear control systems \cite{peitz2020data}. 

However, to control unknown nonlinear systems, factors such as - finite-dimensional approximation, finite samples, and process disturbances - lead to modeling errors \cite{khosravi2022representer}, which is important to consider so as
not to compromise control performance and safety.
To consider explicitly the non-negligible approximation errors and unknown additive disturbances, the works \cite{zhang2022robust, mamakoukas2022robust} apply robust MPC with the Koopman operator on unknown disturbed nonlinear systems. However, robust MPC could result in a conservative closed-loop performance of systems or infeasibility.

Hence we propose a mixed-tube-based data-driven stochastic MPC using the Koopman operator. The reason for using a mixed tube approach is that it allows us to address the deterministic uncertainties (those attributed to the approximation error of listing the system), and the stochastic uncertainty (that inherent to a stochastic dynamic system) \cite{paulson2019mixed}. The deterministic tube part is estimated from the data collected from the evolution of the original and lifted systems, and the stochastic tube part is learned via a DRO approach using the Wasserstein ambiguity set.

Recently, there has been a surge of interest in  distributionally robust control utilizing the Wasserstein ambiguity set. To satisfy state constraints, recent publications \cite{mark2020stochastic,  zhong2021data, coulson2021distributionally, micheli2022data, fochesato2022data, zhong2023tube} address distributionally robust MPC problems associated with the Wasserstein ambiguity set. In these studies, the center of the ambiguity set is determined based on independent and identically distributed (i.i.d.) samples of disturbance sequences. However, due to the inherent characteristics of the equivalent reformulation of Type-1 Wasserstein Distributionally Robust Optimization (DRO), constructing stochastic tubes online in all aforementioned works leads to a substantial increase in computational complexity when the sample size is large. In this paper, we significantly reduce the computational complexity of the online problem by constructing the tube associated with stochastic state constraints offline. Additionally, this approach ensures desirable control properties such as recursive feasibility.


Contribution: \update{In this paper, we propose a general data-driven control framework for nonlinear systems based on Koopman operator theory. The approach combines Koopman-based modeling with a computationally efficient distributionally robust MPC scheme that ensures a balance between tractability, performance, and probabilistic feasibility.  From the theoretical perspective, our approach enables the design of an MPC controller based on measured data using Koopman operators, whilst guaranteeing recursive feasibility and closed-loop robustness in probability (Theorem \ref{th:recursive_feas} and \ref{th:robustness}). In practical terms, our method outperforms existing work \cite{zhang2022robust} by delivering less conservative closed-loop performance (Fig. \ref{fig:compare}). Additionally, we introduce an innovative data-driven technique for constructing the tube corresponding to stochastic errors, which leverages distributionally robust optimization and a Wasserstein ambiguity set (Lemma \ref{lemma:cvar_DRO_reform}). This approach significantly reduces the computational complexity of the online problem compared to \cite{mark2020stochastic,  zhong2021data, coulson2021distributionally, micheli2022data, fochesato2022data, zhong2023tube}}.

The rest of the paper is organized as follows. In Section II, we introduce the control problem and the preliminary Koopman MPC. Section III describes the methodologies applied to this work. In Section IV, we analyze the recursive feasibility and robustness of the closed-loop system. Finally, Section V provides a numerical experiment of a mass-spring system.

\section{PRELIMINARIES}
\update{In this section, we first introduce the control problem formulation. Then, the Koopman theory for controlled dynamical systems and the corresponding lifted disturbed systems constructed accordingly are introduced. Based on the lifted systems, a mixed robust-stochastic SMPC are introduced to address the control problem.}


\subsection{System dynamics, constraints and objective}
We consider the nonlinear time-invariant (LTI) stochastic dynamic system with additive disturbance
\begin{equation}
\label{eq:system}
x^{+}=f_d(x, u) + w_{k}, \quad k \in \mathbb{N}_{\ge 0} 
\end{equation}
with the state $x \in \mathcal{X}$, the control $u \in \mathcal{U}$, and the additive disturbance $w \in \mathbb{W}_w$. Here $\mathcal{X}$ and  $\mathbb{W}_w$ are subsets of $\mathbb{R}^{n_x}$ and $\mathcal{U}$ is a subset of $\mathbb{R}^{n_u}$. \update{Each spaces above are equipped with the Euclidean norm $\| \cdot \|_2$.} We denote the state, input and disturbance realizations at time $k$ with subscript $k$, and it is assumed that the state $x_{k}$ can be measured exactly at time $k$ for all $k \in \mathbb{N}_{\ge 0}$, whereas the system dynamics $f_d: \mathbb{R}^{n_x} \times \mathbb{R}^{n_u}$ are unknown, but following the standard assumption \cite{mamakoukas2020learning}. 
\begin{assumption}
The function $f_d$ is continuously differentiable on the set $\mathcal{X} \times \mathcal{U}$, with constants $L_x$ and $L_u$ denoting the Lipschitz constants with respect to the state and input variables, respectively, \update{i.e. $f_d \in C^1(\mathcal{X} \times \mathcal{U})$, $\| f_d(x,u) - f_d(y,u)\|_2 \le L_x \|x-y\|_2$, and $\| f_d(x,u) - f_d(x,v)\|_2 \le L_u \|u-v\|_2$  $\forall x,y \in \mathcal{X}$ and  $\forall u,v \in \mathcal{U}$.} Furthermore, $(0,0)$ is the equilibrium point of $f_d$, i.e. $f_d(0,0) = 0$.
\end{assumption}
Additionally, each disturbance $w_k$ of the disturbance sequence $\{ w_{k} \}_{k \in \mathbb{N}_{\ge 0}}$ is assumed to be a realization of the corresponding random variable (r.v.) $W_{k}$ from the random process $\{W_{k}\}_{k \in \mathbb{N}_{\ge 0}}$ satisfying the following assumption.

\begin{assumption}[Bounded i.i.d. Random Disturbance]
\label{assump:iid}
All random variables $W_k$ for $k \in \mathbb{N}_{\ge 0}$ from the family of random variables $\{W_{k}\}_{k \in \mathbb{N}_{\ge 0}}$ are assumed to be i.i.d. with an unknown probability distribution $\mathbb{P}_{w}$ and bounded within an unknown polyhedral support $\mathbb{W}_{w} \triangleq \{w \mid H_w w \le h_w\}$.
\end{assumption}


Assuming that the data set $\mathcal{D}:=\{(x_i, u_i, x^{+}_i)\}_{i=1}^{N_d}$ is available, the main target of the paper is proposing a control scheme to track the system state to the tracking point (w.l.o.g. the origin of coordinates) while satisfying individual state chance constraints and hard control input constraints
\updateNew{\begin{subequations}
\begin{align}
\mathbb{P}\left\{[F]_{j} x_{k} \le [f]_{j} \right\} &\ge 1 - [\alpha]_{j}, \quad j \in \mathbb{N}_{1}^{n_F}, \quad k \ge 1,\label{eq:sys_constraints_state}\\
    G u_{k} &\leq g, \quad k \geq 0,\label{eq:sys_constraints_input}   
\end{align}
\end{subequations}}
in the presence of disturbances, where $F \in \mathbb{R}^{n_F \times n_x}, f \in \mathbb{R}^{n_F}, G \in \mathbb{R}^{n_G \times n_u}, g \in \mathbb{R}^{{n_G}}$. Furthermore, we denote $\mathbb{X}=\left\{x \in \mathbb{R}^{n_x}: F x \leq f\right\} \subset \mathcal{X}$ and $\mathbb{U}=\left\{u \in \mathbb{R}^{n_u}: G u \leq g\right\} \subset \mathcal{U}$ as the collection of nominal state and control input constraints, respectively.

\subsection{Koopman-based model}
\update{For controlled systems, $x^{+} = f_d(x,u)$, the Koopman operator can be defined as an operator acting on the function space of extended states. 
Let $\ell(\mathcal{U}):=\{ \left\{u_{i}\right\}_{i=0}^{\infty}\mid u_{i} \in \mathcal{U}\}$ be the space of all infinite control sequences and denote the extended state as $\chi:=\left[\begin{array}{l}x \\ \boldsymbol{u}\end{array}\right]$ and the corresponding dynamical system on the extended state space as $\chi^{+}=F_e(\chi):=\left[\begin{array}{c}f_d(x, \boldsymbol{u}(0)) \\ \Gamma \boldsymbol{u}\end{array}\right]$, where $\Gamma$ is the left shift operator $(\Gamma \boldsymbol{u})(i)=\boldsymbol{u}(i+1)$ and $\boldsymbol{u}(0)$ is the first element of the sequence $\boldsymbol{u} :=\left\{u_{i}\right\}_{i=0}^{\infty}$. 
The associated Koopman operator is then given as $\mathcal{K} \psi(\chi)=\psi(\chi) \circ F_e(\chi)$
where $\psi(\mathcal{X}): \mathcal{X} \times \ell\left(\mathcal{U}\right) \rightarrow \mathbb{R}$ is known as an observable function belonging to the extended observable space $\mathcal{F}$ \cite{korda2018linear}. We denote the vector of observables as $\Psi(\chi) := [\psi_1(x), \dots, \psi_n(x), \boldsymbol{u}(0) ]^{\top}$, the finite-dimensional representation of the Koopman operator associated with the extended state can be acquired from minimizing
$
\sum_{j=1}^{M}\left\|\Psi\left(\chi_j^{+}\right)-\mathcal{K} \Psi\left(\chi_j\right)\right\|_2^2,
$
where $M \in \mathbb{Z}$ is the number of samples. As we are only interested in predicting the lifted state at the next time step, we focus on the first $n$ rows of $\mathcal{K}$ and decompose it into $A \in \mathbb{R}^{n \times n}$ and $B \in \mathbb{R}^{n \times n_u}$, which leads to 
\begin{equation}
\label{eq:lifted_system_learning}
\min _{A, B} \sum_{j=1}^M\left\|\Psi_{1:n}\left(x_j^{+}\right)-A \Psi_{1:n}\left(x_j\right)-B u_j\right\|_2^2.
\end{equation}
We denote the nominal lifted system $s^{+}=A s+B u$, where $\Psi_{1:n}(x) = [\psi_1(x), \dots, \psi_n(x)]^{\top} $ and $s = \Psi_{1:n}(x)$.}

\update{For controlled systems, $x^{+} = f_d(x,u)$, the Koopman operator can be defined as an operator acting on the function space of extended states. 
Let $\ell(\mathcal{U}):=\{ \left\{u_{i}\right\}_{i=0}^{\infty}\mid u_{i} \in \mathcal{U}\}$ be the space of all infinite control sequences and denote the extended state as $\chi:=\left[\begin{array}{l}x \\ \boldsymbol{u}\end{array}\right]$ and the corresponding dynamical system on the extended state space as $\chi^{+}=F_e(\chi):=\left[\begin{array}{c}f_d(x, \boldsymbol{u}(0)) \\ \Gamma \boldsymbol{u}\end{array}\right]$, where $\Gamma$ is the left shift operator $(\Gamma \boldsymbol{u})(i)=\boldsymbol{u}(i+1)$ and $\boldsymbol{u}(0)$ is the first element of the sequence $\boldsymbol{u} :=\left\{u_{i}\right\}_{i=0}^{\infty}$. 
The associated Koopman operator is then given as $\mathcal{K} \psi(\chi)=\psi(\chi) \circ F_e(\chi)$
where $\psi(\mathcal{X}): \mathcal{X} \times \ell\left(\mathcal{U}\right) \rightarrow \mathbb{R}$ is known as an observable function belonging to the extended observable space $\mathcal{F}$ \cite{korda2018linear}. We denote the vector of observables as $\Psi(\chi) := [\psi_1(x), \dots, \psi_n(x), \boldsymbol{u}(0) ]^{\top}$, the finite-dimensional representation of the Koopman operator associated with the extended state can be acquired from minimizing
$
\sum_{j=1}^{M}\left\|\Psi\left(\chi_j^{+}\right)-\mathcal{K} \Psi\left(\chi_j\right)\right\|_2^2,
$
where $M \in \mathbb{Z}$ is the number of samples. As we are only interested in predicting the lifted state at the next time step, we focus on the first $n$ rows of $\mathcal{K}$ and decompose it into $A \in \mathbb{R}^{n \times n}$ and $B \in \mathbb{R}^{n \times n_u}$, which leads to 
\begin{equation}
\label{eq:lifted_system_learning}
\min _{A, B} \sum_{j=1}^M\left\|\Psi_{1:n}\left(x_j^{+}\right)-A \Psi_{1:n}\left(x_j\right)-B u_j\right\|_2^2.
\end{equation}
We denote the nominal lifted system $s^{+}=A s+B u$, where $\Psi_{1:n}(x) = [\psi_1(x), \dots, \psi_n(x)]^{\top} $ and $s = \Psi_{1:n}(x)$.}

\update{In this work, we make the same assumptions for candidate lifting functions $\Psi_{1:n}$ as in \cite{zhang2022robust}.}

\begin{assumption}
\label{assump:lift_func}
$\Psi_{1:n} = (x, \Psi_{n_x:n}(x))$ and $\Psi_{1:n}(0)=0$

\end{assumption}

\update{The first part of the assumption is made to guarantee that $\Psi_{1:n}^{-1}$ exists by selecting $\Psi_{1:n}^{-1}(s)=\left[\begin{array}{ll}I_{n_x} & 0\end{array}\right] s$ such that $\Psi_{1:n}^{-1}(\Psi_{1:n}(x))=x$. The second part is made for the purpose of regulating the origin of coordinates. Such a lifting function can be constructed through coordinate transformation by defining a new function $\Psi_{1:n}(x)=\Psi_{1:n}^{\prime}(x)-\Psi_{1:n}^{\prime}(0)$ for arbitrary $\Psi_{1:n}^{\prime}$. A discussion about the conditions that guarantee a Lyapunov-stable equilibrium in Koopman dynamics can be found in \cite{mamakoukas2020learning}.}

\begin{remark}
    The method proposed in this work is also valid for more general lifting functions. However, for the sake of simplicity, in this work we only discuss the scenario where the lifting functions contain the identity state mapping function.
\end{remark}
\update{
After defining the lifting function $\Psi_{1:n}$, we can acquire a data-driven equivalent Koopman model of \eqref{eq:system} learned from \eqref{eq:lifted_system_learning} by using the data set $\mathcal{D}$, while explicitly considering the modeling error resulting from the Koopman operator and the additive disturbance of \eqref{eq:system} impacting the lifted system:
\begin{equation}
\label{eq:lifted_system}
\left\{\begin{array}{l}
s^{+}=A s+B u+d\left(s, u, w\right) + Dw\\
x=C s,
\end{array}\right.
\end{equation}
where $d\left(s, u, w\right) \in \mathbb{W}_d$ is the modeling error, $w \in \mathbb{W}_w$ is the additive disturbance of \eqref{eq:system}, and $\mathbb{W}_d$ and $\mathbb{W}_w$ are bounded and convex sets containing the origin. The matrix D is by design selected as $\left[\begin{array}{ll}I_{n_x}^{\top} & 0^{\top}\end{array}\right]^{\top}$.  Both sets are learned from the data set $\mathcal{D}$, which is introduced in the previous subsection and the data set with zero-valued inputs $\mathcal{D}_w$, which will be introduced in \updateNew{Section \ref{sec:addtive_dist_est}}. Assumption \ref{assump:lift_func} allows us to select 
$C=\left[\begin{array}{ll}I_{n_x} & 0\end{array}\right]$ such that for some $d \in \mathbb{W}_d$ we have $f_d(x,u) = CA\Psi_{1:n}(x) + CBu + Cd$.}

Here we make assumptions on the observability and stabilizability for the purpose of guaranteeing stability and convergence in Section \ref{sec:proof}.
\begin{assumption}
    The pair $(A, B)$ is stabilizable and the pair $(A, C)$ is observable.
\end{assumption}

\subsubsection{Mixed stochastic-deterministic tube Koopman MPC}
\update{For the purpose of equilibrium \update{regulation}, while satisfying the constraints, MPC method is applied to the lifted system by controlling the predicted states of the future system evolution based on current measurements, i.e. 
given the state measurement $x_k$ at the sample time $k$,  we aim to apply MPC to the lifted state $s_k = \Psi_{1:n}(x_k)$. 
We model the predictions as
\begin{equation}
\label{eq:predicted_lift_system}
\begin{aligned}
s_{i+1 \mid k}&=A s_{i \mid k}+B u_{i \mid k}+d_{i \mid k} + Dw_{i \mid k},\\
s_{0 \mid k}&= s_{k}:=\Psi_{1:n}(x_k),
\end{aligned}
\end{equation}
where  $s_{i \mid k}$ is the uncertain state prediction, $u_{i \mid k}: \mathbb{R}^{n} \rightarrow \mathbb{R}^{m}$ are measurable state feedback functions, $d_{i \mid k} \in \mathbb{W}_{d}$ is the predicted plant-model mismatch $i$ steps ahead of time $k$, and $w_{i \mid k}=W_{k+i}$ as a random variable supported on $\hat{\mathbb{W}}_{w}$. }

\update{
This subsection introduces the general approach of constraint tightening for the lifted nominal system such that it is sufficient to guarantee \eqref{eq:sys_constraints_state} and \eqref{eq:sys_constraints_input}. The technique applied here is to identify a tube around the nominal trajectory, such that the real trajectory resulting from uncertainties could still remain in the tube. Also, the tube is split into static and probabilistic parts to account for the effect of the modeling and additive disturbance respectively. }

\update{We introduce now how to split the lifted states into nominal and uncertain components \cite{paulson2019mixed}. Since there are two sources of uncertainty considered, the predicted state is decomposed into nominal state prediction $\bar{s}_{i \mid k}$, zero-mean stochastic error $e_{i \mid k}$ and deterministic modeling error $\varepsilon_{i \mid k}$:
$$
s_{i \mid k}=\bar{s}_{i \mid k}+e_{i \mid k}+\varepsilon_{i \mid k}.
$$}

\update{By defining the control law $u_{i\mid k}=Ks_{i\mid k} + c_{i\mid k} = K\left(\bar{s}_{i \mid k} + e_{i \mid k}+\varepsilon_{i \mid k}\right)+c_{i \mid k}$, the closed-loop prediction paradigm then must satisfy the following dynamic equations:
\begin{equation}
\label{eq:pred_paradigm}
\begin{aligned}
&\bar{s}_{i+1 \mid k}=\Phi \bar{s}_{i \mid k}+B \bar{u}_{i \mid k}, \quad \bar{s}_{0 \mid k}=s_{k}:=\Psi_{1:n}\left(x_{k}\right) \\
&e_{i+1 \mid k}=\Phi e_{i \mid k}+D w_{i \mid k}, \quad e_{0 \mid k}=0, \\
&\varepsilon_{i+1 \mid k}=\Phi \varepsilon_{i \mid k}+d_{i \mid k}, \quad \varepsilon_{0 \mid k}=d_{0\mid k},
\end{aligned}
\end{equation}
where $\Phi = A + BK$. 
}

\update{
Based on the prediction paradigm \eqref{eq:pred_paradigm}, we can formulate the reachable set of errors in $i$ step from the origin of the system as $R_{i}^{e}=\bigoplus_{j=0}^{i-1} \Phi^{j} D \mathbb{W}_w, \quad R_{i}^{\varepsilon}=\bigoplus_{j=0}^{i-1} \Phi^{j} \mathbb{W}_d$, respectively. We assume in the paper that the pair $(A,B)$ is stabilizable, hence we could find a feedback gain $K$, such that $\Phi$ is strictly stable. Since $\Phi$ is strictly stable, the set sequence $R_{i}^{e}$ and $R_{i}^{\varepsilon}$ converge to the limit $R_{\infty}^{e}$ and $R_{\infty}^{\varepsilon}$, respectively. Both sets are known as robust positively invariant sets, which are defined as the following 
\begin{definition}[Robust positively invariant set]
A set $\Omega$ is called robust positively invariant (RPI) for a system $x_{k+1}=A x_{k}+w_{k}$ with $w_{k} \in W$ for all $k \geq 0$ if, for any $x \in \Omega$, $A x+w \in \Omega$ for all $w \in W$. In other words, $\Omega$ should satisfy $A \Omega \oplus W \subseteq \Omega$.
\end{definition}}

\update{
As next, to guarantee the satisfaction of constraints \eqref{eq:sys_constraints_state} and \eqref{eq:sys_constraints_input}, we consider the following constraints tightening for the lifted systems corresponding to the reachable sets, namely we require that 
\begin{equation}
\label{eq:state_constraints_tighten}
\bar{s}_{i+1 \mid k} \in \bar{S}_{j+1}=\bar{S}_{j} \ominus \Phi^{j} D \hat{\mathbb{W}} \ominus \Phi^{j} \mathbb{D}, \quad j \geq 1 \end{equation}
for $i \ge 0$, where $\mathbb{S}_{1}=\left\{s \in \mathbb{R}^{n}: FCs \leq h-\eta\right\}$ and $\bar{S}_{1}=\mathbb{S}_{1} \ominus \mathbb{D}$. Roughly speaking, the support set of states is tightened by $\eta$ to satisfy the chance constraints and further tightened to guarantee recursive feasibility with respect to the accumulated uncertainty.
Similarly, hard input constraints for the nominal input based on the fact that $ u_{i \mid k}=\bar{u}_{i \mid k}+K e_{i \mid k}+K \varepsilon_{i \mid k}$ can be defined as 
\begin{equation}
\label{eq:input_constraints_tighten}
\bar{u}_{i \mid k} \in \bar{U}_{i}  = \bar{U}_{i-1} \ominus K \Phi^{i} D\hat{\mathbb{W}} \ominus K \Phi^{i} \mathbb{D},  \forall i \geq 1,
\end{equation}
where $\bar{U}_{0}=\mathbb{U}$. The details of the parameter designs for the constraints can be found below in the remark. Throughout the paper we assume that $\Omega$ is closed and contains the origin in the interior.
\begin{remark}
In this work, we consider  constraints $\mathbb{X}$, $\mathbb{S}$ and $\mathbb{U}$ being polytopic. The tightened constraints corresponding to the reachable set of errors, which remain polytopic, can be acquired by solving linear optimization problems. Given the design of $\mathbb{S}_{1}$ mentioned above, $\bar{S}_{1}$ is determined via $\bar{S}_{1} = \left\{s \in \mathbb{R}^{n}: FCs \leq h-\eta - \eta_1\right\}$ where $[\eta_1]_l = \argmax_{d \in \mathbb{D}} I_l^{\top} d$, $\forall l = 1,\dots,n_F$. Here $I_l = [0,\dots,0,1,0,\dots,0]^{\top}$ is a column vector with one at the $l$-th entry and zero for the rest. And $\bar{S}_{j} = \left\{s \in \mathbb{R}^{n}: FCs \leq h-\eta - \sum_{i=1}^{j}\eta_j\right\}$ with $[\eta_1]_l = \argmax_{d \in \mathbb{D}} I_l^{\top} \Phi^{j-1} d + \argmax_{w \in \hat{\mathbb{W}}} I_l^{\top} \Phi^{j-1} D w$ for $j > 1$. The design for the other constraints is similar. More details can be found in \cite[Theorem 8.1]{kouvaritakis2016model}. 
\end{remark}
}

\updateNew{
\begin{remark}
    Given two sets $\mathbb{A},\mathbb{B} \in \mathbb{R}^{n}$, we demonstrate the set operations $\mathbb{A} \oplus \mathbb{B}=\{a+b \mid a \in \mathbb{A}, b \in \mathbb{B}\}$ and
    $\mathbb{A} \ominus \mathbb{B}:=\{a \in \mathbb{A} \mid a+b \in \mathbb{A}, \quad \forall b \in \mathbb{B}\}$ in Figure \ref{fig:min_sum} and \ref{fig:pon_diff}, respectively.
\end{remark}
}

\begin{figure}[h!]
    \centering
    \begin{minipage}[t]{0.45\textwidth}
        \centering
        \includegraphics[width=\linewidth]{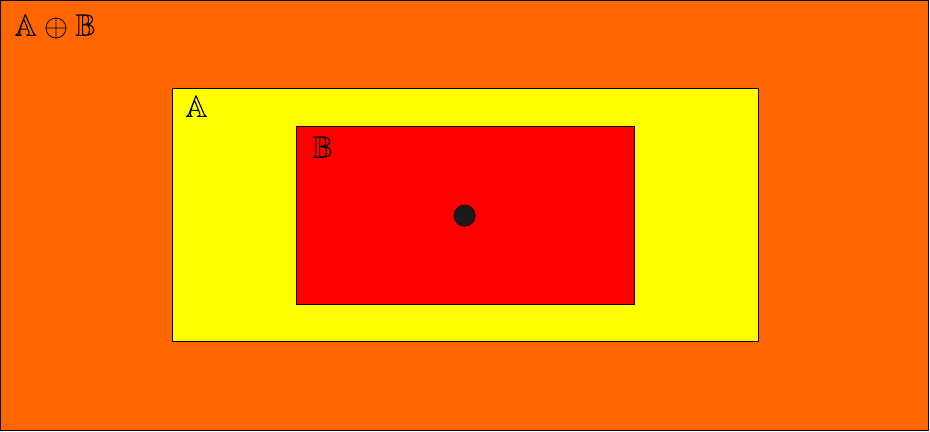}
        \caption{\updateNew{Minkowski sum $\mathbb{A} \oplus \mathbb{B}$}}
        \label{fig:min_sum}
    \end{minipage}
    \hfill
    \begin{minipage}[t]{0.45\textwidth}
        \centering
        \includegraphics[width=\linewidth]{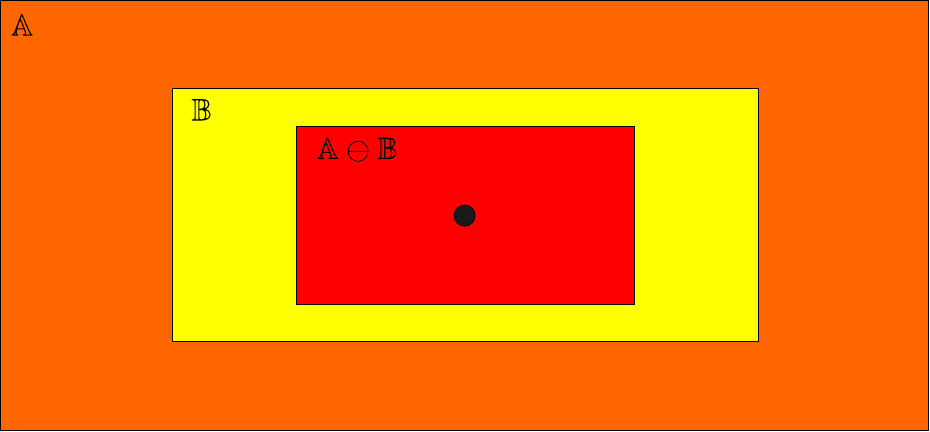}
        \caption{\updateNew{Pontryagin difference $\mathbb{A} \ominus \mathbb{B}$}}
        \label{fig:pon_diff}
    \end{minipage}
\end{figure}

\update{
Now we are ready to introduce the MPC problem to compute $\bar{u}$ and $\bar{s}$. At any time instant $k$, we solve a quadratic optimization with the following objective function
$V\left(\mathbf{c}_{k}\right)=\sum_{i=0}^{N-1}\left\|c_{i \mid k}\right\|_{\Pi}^{2}$,
where $\Pi=R+B^{\top} P B$. Here $P$ is the solution to the Lyapunov equation $P - \Phi^{\top} P \Phi =\bar{Q}+K^{\top} R K,$
, where $\bar{Q}=C^{\top} Q C$. Through  designing the penalty matrix $\Pi$, the finite-horizon costs are equal to the corresponding LQR costs minus a constant term. More details can be found in for example \cite{kouvaritakis2016model}.
Given the objective function, we define the following optimal control problem
\begin{equation}
\label{eq:opt_prob}
\begin{aligned}
&\min _{\boldsymbol{\bar{s}_{k}}, \boldsymbol{\bar{u}_{k}}, \boldsymbol{c}_{k}} V\left(\mathbf{c}_{k}\right) \\
& \text { s.t. } \\
    & \begin{aligned}
    \bar{s}_{0 \mid k} & =s_{k}, \\
    \bar{s}_{i+1 \mid k} &=A \bar{s}_{i \mid k}+B \bar{u}_{i \mid k}, \quad i=0, \ldots, N-1, \\
     \bar{u}_{i \mid k} & =K \bar{s}_{i \mid k} + c_{i \mid k}, \quad i=0, \ldots, N-1\\
    \bar{s}_{i+1 \mid k} & \in \bar{S}_{i+1}, \quad \bar{u}_{i \mid k} \in \bar{U}_{i}, i=0, \ldots, N-1, \\
    \bar{s}_{N \mid k} &\in \bar{S}_{\infty},
\end{aligned}
\end{aligned}
\end{equation}
\updateNew{where $\bar{S}_{\infty}$ is the robust positively invariant set for the lifted system. We will show later in Section \ref{sec:proof} that by designing the constraints of the lift system through tightened tubes around the nominal trajectory, the chance constraints of the original system are satisfied and the optimal control problem remains recursively feasible.}
}



\section{Methodology}
\subsection{Estimating the modeling error and additive disturbance}
In this subsection, we propose methods to estimate additive disturbances and modeling errors. The estimated additive disturbances will be further used for the design of the tightening $\eta$ via a distributionally robust optimization and the modeling errors will be considered to design the reachable sets $R_i^{\varepsilon}$.
\subsubsection{Additive disturbance estimation}
\label{sec:addtive_dist_est}
We first introduce a simple method to estimate additive disturbances from two consecutive states under zero input. Let $\mathcal{D}_w$ be the data set containing $\{x_{i}, x^{+}_{i}\}_{i=1}^{N_w}$, where $x^{+}_{i} := f_d(x_{i},0)+w_{i}$ is the successor state measurement of the autonomous real system \eqref{eq:system}, i.e.  under zero input and an unknown disturbance $w_{i}$. 
We estimate the disturbance through $\hat{w}_i = x^{+}_{i}$. The estimation error could be \updateNew{bounded} by a small value via the following lemma.
\begin{lemma}
    \label{lemma:estimation_error}
    Let Assumption 1 hold, then $\|\hat{w}_i - w_{i}\| \le L_x \|x_i\|$.
\end{lemma}
\begin{proof}
    $$
\begin{aligned}
    \| \hat{w}_i - w_i\| & = \| x^{+}_i - w_i\|\\
                         &= \| f(x_i,0)\| \\
                         &= \| f(x_i,0) + f(0,0) - f(0,0)\| \\
                         & \le \| f(x_i,0) + f(0,0)\| + \|f(0,0)\| \\
                        & \le L_x \|x_i\|.
\end{aligned}
$$
\end{proof}
Through the proposed way of estimating additive disturbances, if $x_i$ is close enough to the origin, i.e. $\|x_i \| \le \epsilon_{x}$, the estimated disturbance can be represented by $\hat{w}_i = x^{+}_{i}$ and the error is bounded via $\|\hat{w}_i - w_{i}\| \le L_x\epsilon_{x}$. In this fashion, the estimation error is small if the Lipschitz constant or the distance $\varepsilon_x$ is small. We define the hyper-cubic support set of $\{\hat{w}\}_{i = 1}^{N_w}$ as $\hat{\mathbb{W}}_{w}$, i.e. $\{\hat{w}\}_{i = 1}^{N_w} \subseteq \hat{\mathbb{W}}_{w}$. We will use the empirical distribution of $\{\hat{w}\}_{i = 1}^{N_w}$ to construct a stochastic tube with a finite sample guarantee explicitly considering the estimation error.

As the support set $\hat{\mathbb{W}}_{w}$ of the estimated additive disturbance is an outer approximation of the true support set of additive disturbances, we make the following assumption.

\begin{assumption}
    We assume $\mathbb{W}_w \subseteq \hat{\mathbb{W}}_{w}$.
\end{assumption}


\subsubsection{Modeling error estimation}
We propose to estimate the total effect of modeling error $d$ and additive disturbances, i.e. $\bar{w} := d + Dw$, via a quantifiable hyper-cubic set. For the given data set $\mathcal{D}$, we calculate the total modeling error via 
\begin{equation}
\label{eq:modelling_error}
\bar{w}_i : = \Psi_{1:n}(x_i^{+}) - (A \Psi_{1:n}(x_i) + B u_{i}).
\end{equation} For the original system, the error is determined via $\bar{w}_i : = x_i^{+} - C (A \Psi_{1:n}(x_i) + B u_{i})$, which consists of the first $n_x$ element of $\bar{w}_i$. We learn a hyper-cube hull of $\{\bar{w}_i\}_{i=1}^{N_d}$ via a hyper-cubic set, which is denoted as $\bar{\mathbb{W}}$. After taking the Pontryagin difference of $\bar{\mathbb{W}}$ and $D\hat{\mathbb{W}}_{w}$, the support set of $d$ is acquired via $\mathbb{D} := \bar{\mathbb{W}} \ominus D\hat{\mathbb{W}}_{w}$, which is a hyper-cubic set as well.

Furthermore, we quantify the required number of data points to guarantee that the probability that $\bar{\mathbb{W}}$ containing all the errors is great than $1-\varepsilon_{h}$ with confidence $1-\delta_{h}$. The technique applied here is Hoeffding’s Inequality. We define the indicator function
$$
\mathbb{I}_{\bar{\mathbb{W}} }\left(\bar{W}\right):= \begin{cases}1 & \text { if} \quad \bar{W} \in \bar{\mathbb{W}}  \\ 0 & \text { otherwise }\end{cases}
$$
to show if the realization $\bar{W}$ is in the set $\bar{\mathbb{W}}$.
\begin{definition}[Hoeffding’s Inequality]
Let $\mathbb{I}_{\bar{\mathbb{W}} }\left(\bar{W}_{i}\right)$ with $i=1, \ldots, N_{w}$ be $N_{w}$ independent random variables with $0 \leq \mathbb{I}_{\bar{\mathbb{W}} }\left(\bar{W}_{i}\right) \leq 1$, and $\mathbb{E}\left[\mathbb{I}_{\bar{\mathbb{W}} }\left(\bar{W}_{i}\right)\right]=\mu$. Then, for any $\epsilon_{h} > 0$,
$$
\mathbb{P}\left[|\bar{\mu}-\mu| \geq \varepsilon_{h}\right] \leq 2 \exp \left(-2 N_{p} \varepsilon_{h}^{2}\right),
$$ where $\bar{\mu}=N_{w}^{-1} \sum_{i} \mathbb{I}_{\bar{\mathbb{W}} }\left(\bar{W}_{i}\right)$.
\end{definition}
\begin{lemma}
\label{lemma:finite_sample_modelling_error}
    Let Assumptions 2, 4 hold and $\bar{w}_{i}, \forall i = 1 \dots N_d$ be derived from \eqref{eq:modelling_error} corresponding to data set $\mathcal{D}:=\{(x_i, u_i, x^{+}_i)\}_{i=1}^{N_d}$ with $N_{d} \ge \frac{- \ln{(\delta_{h}/2})}{2 \epsilon_h^{2}}$ for the some $\varepsilon_h$ and $\delta_h$. Also, let the set $\bar{\mathbb{W}}$ be a hyper-cubic hull of $\{\bar{w}_i\}_{i=1}^{N_d}$. Then, we have $\mathbb{E}\left[\mathbb{I}_{\bar{\mathbb{W}} }\left(\bar{W}\right)\right] = \mathbb{E}\left[\mathbb{I}_{\bar{\mathbb{W}} }\left(\bar{W}_{i}\right)\right] \ge 1 - \epsilon_h$ with confidence $1-\delta_h$.
\end{lemma}
\begin{proof}
    This directly comes from the definition of Hoeffding's inequality with $\mathbb{I}_{\bar{\mathbb{W}} }\left(\bar{W}_{i}\right) = 1$ for all $i = 1, \dots, N_w$.
\end{proof}

\subsection{Chance constraints design}
We are now ready to introduce the design of $\eta$ such that the chance constraints can be satisfied.
\subsubsection{CVaR constraints}
It is well-known that constraints $\mathbb{P}(F(u, \xi) \leq 0) \geq 1-\varepsilon$, where $F: \mathbb{R}^{n_u} \times \mathbb{R}^{n_{\xi}} \rightarrow \mathbb{R}$, have a non-convex feasibility set even if $X$ is convex and $F$ is convex in $u$ for every $\xi$ \cite{hota2019data}.
In this paper, we consider the approximation known as Conditional value-at-risk (CVaR). 

We replace the chance constraint with a convex conservative approximation and denote the left-hand side of the inequality as CVaR of the random variable $Z$ with distribution $\mathbb{P}$ at level $\varepsilon$: $\operatorname{CVaR}_{1-\varepsilon}^{\mathbb{P}}(Z):=\inf _{t \in \mathbb{R}}\left[\varepsilon^{-1} \mathbb{E}_{\mathbb{P}}\left[(Z+t)_{+}\right]-t\right] .$ In the next subsection we introduce a data-driven method to identify the level $\varepsilon$ via distributionally robust optimization.
\subsubsection{Distributionally robust optimization and Ambiguity Sets}
\label{sec:DRO_ambiguity}
Distributionally robust optimization 
\begin{equation}
\label{eq:DRO_prototype}
\inf _{x \in \mathbb{X}} \sup _{\mathbb{Q} \in \mathcal{P}} \mathbb{E}_{\mathbb{Q}}[h(x, \xi)]
\end{equation}
is an optimization model which utilizes partial information about the underlying true probability distribution $\mathbb{P}_{\text{true}}$ of the random variables $\xi$ in a stochastic model. To characterize the partial information about the true distribution, we define an ambiguity set $\mathcal{P}$ \cite{mohajerin2018data}, which contains a family of probability measures on the measurable space $(\Omega, X)$. In this paper, Wasserstein ambiguity sets will be considered, as solving the DRO problem with the Wasserstein ambiguity set enables us to quantify the value of linear functions of random variables without assuming the exact knowledge of the true probability distribution. We solve such a data-driven DRO problem offline to acquire the stochastic tube, namely the tightening accounting for the chance constraints based on the estimated disturbances. 

A Wasserstein ambiguity set is modeled as a Wasserstein ball centered at a discrete empirical distribution. The Wasserstein ball is a discrepancy-based model wherein the distance  between probability distributions on the probability distribution space and the empirical distribution at the ball center is described by the Wasserstein metric, which is defined below.

\begin{definition}[Wasserstein Metric \cite{ambrosio2005gradient}]
\label{def:Wassertein_metric}
The \emph{Wasserstein metric} of order $p \ge 1$ is defined as $d_w: \mathcal{M}(\mathbb{W}_{\xi}) \times \mathcal{M}(\mathbb{W}_{\xi}) \rightarrow \mathbb{R}$ for all distribution $\mathbb{Q}_1, \mathbb{Q}_2 \in \mathcal{M}(\mathbb{W}_{\xi})$ and arbitrary norm on $\mathbb{R}^{n_\xi}$:
$$d^{p}_{w}\left(\mathbb{Q}_{1}, \mathbb{Q}_{2}\right):=\inf_{\Pi} \int_{\mathbb{W}_{\xi}^{2}}\left\|\xi_{1}-\xi_{2}\right\|^{p} \Pi\left(\mathrm{d} \xi_{1}, \mathrm{~d} \xi_{2}\right),$$
where $\Pi$ is a joint distribution of $\xi_{1}$ and $\xi_{2}$ with marginals $\mathbb{Q}_{1}$ and $\mathbb{Q}_{2}$ respectively, and $\mathcal{M}(\mathbb{W}_{\xi})$ is the set of all probability distributions supported on $\mathbb{W}_{\xi}$ with a finite expectation.
\end{definition}

We only consider a type-1 Wasserstein metric in this work. Specifically, we define the ambiguity set $\mathcal{P}_{N_s}$ centered at the empirical distribution leveraging the Wasserstein metric
\updateNew{
\begin{equation}
\mathcal{P}_{N_s}:=\left\{\mathbb{Q} \in \mathcal{M}(\mathbb{W}_{\xi}): d_{w}\left(\hat{\mathbb{P}}_{N_s}, \mathbb{Q}\right) \leq \varepsilon\right\},
\label{eq:Wasserstein_ball}
\end{equation}
}
which specifies the Wasserstein ball with radius $\varepsilon>0$ around the discrete empirical probability distribution $\hat{\mathbb{P}}_{N_s}$. The empirical probability distribution  $\hat{\mathbb{P}}_{N_s}:=\frac{1}{N_s} \sum_{l=1}^{N_s} \delta_{\hat{\xi}^{(i)}}$ is the mean of $N_s$ Dirac distributions which concentrates unit mass at the disturbance realization $\hat{\xi}^{(i)} \in \mathbb{W}_{\xi}$. We denote the training set of offline collected realizations $\xi$ as $\mathcal{D}^{N_s}_{\xi}:=\left\{\hat{\xi}^{(i)}\right\}_{i = 1}^{N_s} \subset \mathbb{W}_{\xi}$, which contains $N_s$ observed disturbance realizations.

After defining the Wasserstein ball, we will introduce the following Lemma, which will be applied to solve a DRO problem with a Wasserstein ambiguity set to derive the tightened nominal state constraints.

\begin{lemma}
\label{lemma:cvar_DRO_reform}
Given a compact support set $\mathbb{W}_{\xi}=\left\{\xi \in \mathbb{R}^{n_{\xi}} \mid H \xi \leq h\right\}$, $N$ vectors $a_i \in \mathbb{R}^{n_{\xi}}$, and $N_s$ realization of $\xi$, the optimization problem
\updateNew{
\begin{equation}
\label{eq:DRO_CVaR_proto}
\begin{array}{ll}
\min_{\eta} & \eta \\
\text { s.t. } & \sup_{\mathbb{P} \in \mathcal{P}_{N_s}} \operatorname{CVaR}_{1-\alpha}^{\mathbb{P}}(a^{\top} \xi - \eta) \le 0
\end{array}
\end{equation}
}
is equivalent to the following tractable convex optimization problem 
\begin{equation}
\label{eq:CVaR_DRO_reform}
\begin{array}{ll}
\min_{\eta} & \eta\\
\text { s.t. } & \lambda \theta+\frac{1}{N} \sum_{l=1}^{N} s_{l} \leq t \alpha, \\
& \left(-\eta + t+ (a^{\top}-\gamma_{l}^{\top}H) \widehat{\xi}_{l}+\gamma_{i l}^{\top} h\right)_{+} \leq s_{l}, \\
& \left\|a - H^{\top} \gamma_{l}\right\| \leq \lambda, \\
& \eta \ge 0, t \in \mathbb{R}, \lambda \geq 0, \gamma_{l} \geq 0, l \in \mathbb{N}_{1}^{N_s}
\end{array}
\end{equation}
\end{lemma}
\begin{proof}
See \cite[Prop. V.1]{hota2019data} with $F(x, \xi):=a^{\top} \xi - \eta$.

\end{proof}
\updateNew{Using $N_s$ offline realizations of the disturbance, the convex optimization problem \eqref{eq:CVaR_DRO_reform} enables us to determine the tightening $\eta$ associated with the chance constraints \eqref{eq:sys_constraints_state} without requiring the exact knowledge of the disturbance’s true underlying probability measure.}

\subsubsection{Recursively feasible tubes for mixed uncertainty}
\label{sec:recursive_feas_design}
Now, we are at the stage of introducing the methods we applied to tighten the constraints for nonlinear states $\bar{s}_{i\mid k}$ and inputs $\bar{u}_{i\mid k}$ such that the constraints \eqref{eq:sys_constraints_state} and \eqref{eq:sys_constraints_input} can be satisfied. We denote here the chance constraints for the lifted system as 
\begin{equation}
\label{eq:lift_state_constraint1}
\mathbb{P} \{[F_{lift}]_j s_k \le [f_{lift}]_j \} \le 1 - [\alpha]_j.
\end{equation}
Notice that with the relation $x = Cs$, we have $Fx = FCs$; hence we select $F_{lift} = FC$ and $f_{lift} = f$. 

In order to guarantee recursive feasibility, we need to guarantee: (1) the satisfaction of the lifted constraints implies the satisfaction of the original constraints; (2) the lifted constraints are satisfied recursively under all possible errors.

The first requirement can be guaranteed if we could guarantee \updateNew{$\inf_{\mathbb{P} \in \mathcal{P}_{N_s}}\mathbb{P}\{\sup_{d_{i\mid k}}[F]_j Cs_{i+1 \mid k} \le [f]_j \mid s_{i \mid k}\} \ge 1 - [\alpha]_j$}, which means that for the worst-case additive disturbance estimation and modeling error, the chance constraints of the original system \eqref{eq:system} is still satisfied. Hence, we define the nominal constraint for the first predicted state as $\bar{S}_1=\mathbb{S}_1 \ominus \mathbb{W}_d$, where $\mathbb{S}_1=\left\{s \in \mathbb{R}^{n_s}: [F]_jC s \leq f-[\eta]_j\right\}$ and each element of $\eta$ is determined through solving 
\updateNew{
$$
\begin{array}{ll}
\min_{[\eta]_{j}} &  [\eta]_j \\
\text { s.t. } & \sup_{\mathbb{P} \in \mathcal{P}_{N_s}} \operatorname{CVaR}_{1-[\alpha]_j}^{\mathbb{P}}([F]_j^{\top} w - [\eta]_j) \le 0.
\end{array}
$$
}


To guarantee the second requirement, we consider a sufficient condition guaranteeing that \eqref{eq:lift_state_constraint1} is satisfied for all $k \ge 0$ can be formulated as below
\begin{equation}
\label{eq:lift_state_constraint1_sufficient}
\mathbb{P}\left\{[F]_jC s_{k+1} \leq [f]_j \mid s_{k}\right\} \geq 1-[\alpha]_j, \quad j \in \mathbb{N}_{1}^{n_F}, \quad k \geq 0. \end{equation}
This sufficient condition guarantees that \eqref{eq:lift_state_constraint1} holds for all reachable $s_k$, i.e. \begin{equation}
 \label{eq:lift_state_constraint1_sufficient2}   
\mathbb{P}\left\{[F]_jC s_{k+1} \leq [f]_j\right\}\\
=\int \underbrace{\mathbb{P}\left\{[F]_jC s_{k+1} \leq [f]_j \mid s_{k}\right\}}_{\geq 1-[\alpha]_j, \forall s_{k}} d F_{s_{k}}\left(s_{k}\right),\end{equation}
where $d F_{s_{k}}\left(s_{k}\right)$ is the probability measure for $s_{k}$, which integrates to one based on the axioms of probability. \updateNew{Since the integral of the probability measure over the whole domain is $1$, Eq. \eqref{eq:lift_state_constraint1_sufficient2} implies that $\mathbb{P}\left\{[F]_jC s_{k+1} \leq [f]_j\right\} \ge 1 - [\alpha]_j$.}

\subsubsection{Measure Concentration for shifted samples}
To justify using the estimated additive disturbance for constraint tightening through solving a DRO, we show that based on the Lipschitz constant $L_x$ a well-selected ball radius centered at the empirical distribution of estimated disturbances could contain the true distribution of additive disturbance.

\begin{assumption}
\label{assump:light_tailed_distr}
There exists an exponent $a>1$ such that 
\updateNew{
\begin{equation}
\mathbb{E}_{\mathbb{P}}\left[\exp \left(\| \xi\|^{a}\right)\right]=\int_{\Xi} \exp \left(\|\xi\|^{a}\right) \mathbb{P}(\mathrm{d} \xi)<\infty
\end{equation}
}
\end{assumption}


\begin{lemma}[Distance between discrete distributions]
\label{lemma:Wass_dist_emp}
For two empirical distributions $\hat{\mathbb{P}}_{N_1} = \frac{1}{N_1} \sum_{i=1}^{N_1} \delta_{\hat{\xi}^{(i)}}$ and $\hat{\mathbb{P}}_{N_2} = \frac{1}{N_2} \sum_{j=1}^{N_2} \delta_{\hat{\xi}^{(j)}}$, the Wasserstein distance between $\hat{\mathbb{P}}_{N_1}$ and $\hat{\mathbb{P}}_{N_2}$ is 
\begin{equation}
\label{eq:Wass_dist_emp}
\begin{aligned}
d_{w} (\hat{\mathbb{P}}_{N_1}, \hat{\mathbb{P}}_{N_2}) = \inf_{M} & \sum_{j = 1 }^{N_2}\sum_{i =1}^{N_1} \|\hat{\xi}^{(i)} - \hat{\xi}^{(j)} \| M_{ij}\\
s.t. & \sum_i M_{ij} = \frac{1}{N_2}, \quad \sum_j M_{ij} = \frac{1}{N_1}.
\end{aligned}
\end{equation}
Let $\mathbb{P}_{N_1}$  and $\mathbb{P}_{N_2}$ both with $N$ sample. If on the same finite support set, the maximal distance $\max\|\hat{\xi}^{(i)} - \hat{\xi}^{(j)} \|$ is $\varepsilon_{\xi}$ for all $\forall i,j \in \mathbb{Z}_{1}^{N}$, then the maximal Wasserstein distance is $\varepsilon_{\xi}$.
\end{lemma}
\begin{proof}
The maximal Wasserstein distance between two empirical distributions is derived from 
$$
\begin{aligned}
\displaystyle\max_{\hat{\xi}^{(i)},\hat{\xi}^{(j)} \in \mathbb{W}_{\xi}} \eqref{eq:Wass_dist_emp} & \le \inf_{M} \max_{\hat{\xi}^{(i)},\hat{\xi}^{(j)} \in \mathbb{W}_{\xi}}  \sum_{j = 1 }^{N_2}\sum_{i =1}^{N_1} \|\hat{\xi}^{(i)} - \hat{\xi}^{(j)} \| M_{ij} \le \varepsilon_{\xi}
\end{aligned}
$$
The first inequality is due to the property of the minimax problem and the second inequality happens because the maximal distance is bounded by $\varepsilon_{\xi}$ and probability masses sum up to one.
\end{proof}
With the help of Lemma \ref{lemma:Wass_dist_emp}, we denote the support sets of the estimated and true disturbance as $\mathbb{W}_{w}$ and $\mathbb{W}_{w} \oplus L_x \varepsilon_x\mathbb{W}_c$, where  $\mathbb{W}_c$ is a unit circle of the norm $\|\cdot \|$. From Lemma \ref{lemma:Wass_dist_emp} we know that $\epsilon_{max}$ is $L_x \varepsilon_x$. Hence, we derive the following theorem for finite sample guarantee of estimated disturbances.

\begin{theorem}
\label{th:finite_sample}
If Assumptions 1, 4, 5 are satisfied, \update{then there exist constants $c_1$ and $c_2$ such that 
with the sample number $N \geq \frac{\log \left(c_{1} \beta^{-1}\right)}{c_{2}}$} the Wasserstein ball with radius 
$\varepsilon_{N}(\beta):= 
\left(\frac{\log \left(c_{1} \beta^{-1}\right)}{c_{2} N}\right)^{1 / \max \{n_x, 2\}} + L_x\varepsilon_x $
contains the unknown data-generating distribution $\mathbb{P}$ with confidence with confidence $1-\beta$ for some prescribed $\beta \in(0,1)$.
\end{theorem}
\begin{proof}
    \update{By \cite[Th 3.4]{mohajerin2018data} we have the Wasserstein ball with radius $\varepsilon_N:=\left(\frac{\log \left(c_{1} \beta^{-1}\right)}{c_{2} N}\right)^{1 / \max \{n_x, 2\}}$ centered at the perfect empirical measure contains $\mathbb{P}$ with confidence $1-\beta$. However, the empirical measure in this work is estimated with the maximal error $L_x \varepsilon_x$ in the corresponding Wasserstein distance. Hence, by Lemma \ref{lemma:estimation_error} and Lemma \ref{lemma:Wass_dist_emp}, $\mathbb{P}$ is contained in the ambiguity set with an enlarged radius $L_x \varepsilon_x$. }
\end{proof}

\section{Recursive feasibility, Stability and Convergence}
\label{sec:proof}
We provide theoretical results about recursive feasibility and stability in this section.

\begin{theorem}[Recursive feasibility]
\label{th:recursive_feas}
If Assumptions 1-5 hold and the optimization problem \eqref{eq:opt_prob} is feasible at the initial time $k = 0$, then it is recursively feasible at all time $k \in \mathbb{N}_{1}^{\infty}$ with probability $1-\beta- \delta_h$. Also the chance constraints \eqref{eq:sys_constraints_state} and \eqref{eq:sys_constraints_input} are satisfied at each $k \in \mathbb{N}_{1}^{\infty}$.

\end{theorem}
\begin{proof}
\update{
This proof follows the tube designs in Section \ref{sec:recursive_feas_design} and the proof in \cite[Theorem 1]{paulson2019mixed}.}

\update{
We first prove the recursive feasibility of the optimization problem \eqref{eq:opt_prob} of the lifted system, then we show the chance constraints of the original system can be satisfied.}

\update{
First, we assume that the optimization problem \eqref{eq:opt_prob} is solvable at time $k$; hence, we denote the optimal values of optimization problem as $\boldsymbol{\bar{s}_{k}}^{\star} = \left(s_{0 \mid k}^{\star}, s_{1 \mid k}^{\star}, \ldots, s_{N \mid k}^{\star}\right), \boldsymbol{\bar{u}_{k}}^{\star} = \left(u_{0 \mid k}^{\star}, u_{1 \mid k}^{\star}, \ldots, u_{N-1 \mid k}^{\star}\right),  \mathbf{c}_{k}^{\star}=\left(c_{0 \mid k}^{\star}, c_{1 \mid k}^{\star}, \ldots, c_{N-1 \mid k}^{\star}\right)$. The optimal nominal state profile and input profile are then given by
$$
\begin{array}{cc}
\bar{s}_{i+1 \mid k}^{\star} =& \Phi \bar{s}_{i \mid k}^{\star}+B c_{i \mid k}^{\star}, \quad \bar{s}_{0 \mid k}^{\star}=s_{k}\\
\bar{u}_{i \mid k}^{\star} =& K \bar{s}_{i\mid k}^{\star} + \bar{c}_{i\mid k}^{\star}
\end{array}
$$
We also denote $\mathbf{c}_{k+1}^{c}: = \left(c_{0 \mid k+1}^{c}, \ldots, c_{N-2 \mid k+1}^{c}, c_{N-1 \mid k+1}^{c}\right) =\left(c_{1 \mid k}^{\star}, \ldots,  c_{N-1 \mid k}^{\star}, 0 \right)$ as the candidate solution at time $k+1$ and let $u_{i \mid k+1}^{c}, i \geq 0$ and $s_{i \mid k+1}^{c}, i \geq 0$ denote the corresponding nominal input and state predictions with $s_{0 \mid k+1}^{c}=s_{k+1}$, respectively. Also, due to the selection of observables, we have $s_{k+1} = \Psi_{1:n}(f_d(x_k,u_k) + w_k)$. Hence, for the lifted system, we have the following relation for the state random variable at next time $k+1$:  $s_{k+1}=A s_{k}+B u_{k}+d_{k}+D \hat{w}_{k}$ for some $d_k \in \mathbb{W}_d$ and $u_{k}=\bar{u}_{0 \mid k}^{\star}=K s_{k} + c_{0 \mid k}^{\star}$ . We could therefore write $\bar{s}_{0 \mid k+1}^{c}$ in terms of $\bar{s}_{1 \mid k}^{\star}$:
$$
\label{eq:recursive_step0}
\begin{aligned}
\bar{s}_{0 \mid k+1}^{c} = s_{k+1}&=A s_{k}+B u_{k}+d_{k}+ D \hat{w}_{k}, \\
&=As_k + BKs_k + Bc^{\star}_{0\mid k} + d_k + D\hat{w}_k\\
&=(A+B K) s_{k}+Bc_{0 \mid k}^{\star}+d_{k}+D\hat{w}_{k}, \\
&=\bar{s}_{1 \mid k}^{\star}+d_{k}+D \hat{w}_{k}.
\end{aligned}
$$}

\update{
By definition above, the candidate nominal state with one step prediction is $\bar{s}_{1 \mid k+1}^{c}=\Phi \bar{s}_{0 \mid k+1}^{c}+B c_{0 \mid k+1}^{c}$. After substituting the one step predicted candidate nominal state by \eqref{eq:recursive_step0}, we can derive
$$
\begin{aligned}
\bar{s}_{1 \mid k+1}^{c} &= \Phi \bar{s}_{0\mid k+1}^{c} + B \bar{u}_{0 \mid k+1}^{c}\\
&=\Phi\left(\bar{s}_{1 \mid k}^{\star}+d_{k}+D \hat{w}_{k+1}\right)+B c_{1 \mid k}^{\star} \\
&=\left(\Phi \bar{s}_{1 \mid k}^{\star} + B c_{1 \mid k}^{\star}\right)+\Phi d_{k}+\Phi D\hat{w}_{k+1}, \\
&=\bar{s}_{2 \mid k}^{\star}+\Phi d_{k}+\Phi D \hat{w}_{k+1} .
\end{aligned}
$$
By induction, we will have the following relation for each candidate nominal states for the prediction step $i \in \mathbb{N}_{0}^{N-1}$:
$$
\bar{s}_{i \mid k+1}^{c}=\bar{s}_{i+1 \mid k}^{\star}+\Phi^{i} d_{k}+\Phi^{i} D\hat{w}_{k+1}, \quad \forall i \in \mathbb{N}_{0}^{N-1}
$$
Similarly,  the candidate nominal input $\bar{u}_{i \mid k+1}^{c}=K \bar{s}_{i \mid k+1}^{c}+c_{i \mid k+1}^{c}$ can be reformulated as 
$$
\bar{u}_{i \mid k+1}^{c}=\bar{u}_{i+1 \mid k}^{\star}+K \Phi^{i} d_{k}+K \Phi^{i} D \hat{w}_{k+1}, \quad \forall i \in \mathbb{N}_{0}^{N-1}.
$$
Based on the construction of reachable sets in \eqref{eq:state_constraints_tighten} and \eqref{eq:input_constraints_tighten} \updateNew{and due to Lemma \ref{lemma:finite_sample_modelling_error} and Theorem \ref{th:finite_sample},  with probability $1-\beta- \delta_h$,} we have 
$$
\begin{aligned}
&\bar{s}_{i+1 \mid k}^{\star} \in \bar{S}_{i+1} \Rightarrow \bar{s}_{i \mid k+1}^{c} \in \bar{S}_{i}, \quad \forall i \in \mathbb{N}_{1}^{N-1}, \\
&\bar{u}_{i+1 \mid k}^{\star} \in \bar{U}_{i+1} \Rightarrow \bar{u}_{i \mid k+1}^{c} \in \bar{U}_{i}, \quad \forall i \in \mathbb{N}_{0}^{N-2}.
\end{aligned}
$$
Next, we show that the tightened constraints are also met at the end of the horizon. Since we have $\bar{s}_{N \mid k}^{\star} \in \bar{S}_{N}$, it guarantees that $\bar{s}_{N-1 \mid k+1}^{c}=\bar{s}_{N \mid k}^{\star}+\Phi^{N-1} d_{k}+\Phi^{N-1} D \hat{w}_{k+1} \in \bar{S}_{N-1}$. Similarly, we have $\bar{u}_{N-1 \mid k+1}^{c} \in \bar{U}_{N-1}$.}

\update{
Lastly, we prove that the terminal constraint is satisfied at time $k+1$, i.e. $\bar{s}_{N \mid k+1}^{c} \in \bar{S}_{\infty}$. We formulate the terminal predicted nominal state at $k+1$ as
$$
\begin{aligned}
\bar{s}_{N \mid k+1}^{c}=& A \bar{s}_{N-1 \mid k+1}^{c}+B \bar{u}_{N-1 \mid k+1}^{c}, \\
=& A\left(\bar{s}_{N \mid k}^{\star}+\Phi^{N-1} d_{k}+\Phi^{N-1} D \hat{w}_{k+1}\right) \\
&+B\left(K \bar{s}_{N \mid k}^{\star}+K \Phi^{N-1} d_{k}+K \Phi^{N-1} D \hat{w}_{k+1}\right) \\
=&(A+B K)\left(\bar{s}_{N \mid k}^{\star}+\Phi^{N-1} d_{k}+\Phi^{N-1} D \hat{w}_{k+1}\right) \\
=& \Phi \bar{s}_{N \mid k}^{\star}+\Phi^{N} d_{k}+\Phi^{N} D \hat{w}_{k+1}
\end{aligned}
$$
in the terminal constraint due to the assumption that $$
\Phi  \bar{S}_{\infty} \oplus \Phi^{N} \mathbb{D} \oplus \Phi^{N} D   \mathbb{W} \subseteq \bar{S}_{\infty}.
$$}
\end{proof}

\begin{remark}
\update{
In the classical approach to guaranteeing recursive feasibility through constraint tightening, the design typically involves two steps (see, for example, \cite{paulson2019mixed}). The first step determines the tightening required to satisfy chance constraints, often formulated as a value-at-risk problem based on a given probability distribution of the uncertainty. In our case, this step is replaced by a CVaR-based distributionally robust optimization, which may fail to identify sufficient tightening with probability at most $\beta$; this happens when the true distribution lies outside the Wasserstein ball. The second step involves additional tightening to ensure recursive feasibility under the worst-case realization of uncertainty. Here, the support set is estimated in such a way that it may fail to outer approximate the true support with probability at most $\delta_h$. Combining both sources of uncertainty, recursive feasibility is guaranteed with probability at least $1-$ $\beta-\delta_h$.}
\end{remark}

After given the recursive feasibility guarantee, we state below the robustness results for the closed-loop nonlinear system. 

\begin{theorem}[Closed-loop Robustness]
\label{th:robustness}
If Assumptions 1-5 hold, then it holds with probability $1-\beta- \delta_h$ that: (a): $\lim _{k \rightarrow \infty} c_{0 \mid k}^{\star}=0$ (b): the state $s$, $x$ and the control $u$ of the closed-loop system \eqref{eq:system} with $u_k = K s_k + c_{0\mid k}^{*}$ are such that $s_{k} \rightarrow \left\{0\right\} \oplus R_{\infty}:=R_{\infty}^{e} \oplus R_{\infty}^{\varepsilon}$ and $x_{k} \rightarrow \left\{0\right\} \oplus CR_{\infty}$ as $k \rightarrow+\infty$
\end{theorem}

\begin{proof}

 Based on Theorem \ref{th:finite_sample} and \ref{lemma:finite_sample_modelling_error}, we have at least probability $1-\beta- \delta_h$ to construct outer-approximated stochastic and deterministic tubes. Then, with probability $1-\beta- \delta_h$, the following statements hold.
 We first introduce the optimal cost function at sampling time $k$: $J_k \triangleq V^{\star}\left(x_k\right)=\sum_{i=0}^{N-1}\left\|c_{i \mid k}^{\star}\right\|_{\Pi}^2$. Then, we define the candidate control sequence $
\mathbf{c}_{k+1}^c=\left(c_{1 \mid k}^{\star}, \ldots, c_{N-1 \mid k}^{\star}, 0\right)$. The cost associated with the candidate control sequence is then:
$$
\begin{aligned}
V^c\left(x_{k+1}\right) & =\sum_{i=0}^{N-1}\left\|c_{i \mid k+1}^c\right\|_{\Pi}^2,\\
& =\sum_{i=1}^{N-1}\left\|c_{i \mid k}^{\star}\right\|_{\Pi}^2, \\
& =J_k-\left\|c_{0 \mid k}^{\star}\right\|_{\Pi}^2 ,
\end{aligned}
$$
where $c_{0 \mid k}^{\star}$ depends on the initial state. As the candidate is suboptimal; therefore, we have the optimal solution $J_{k+1} \leq V^c\left(x_{k+1}\right)$. By taking the difference between two optimal costs, we have $J_{k+1}-J_k \leq-\left\|c_{0 \mid k}^{\star}\right\|_{\Pi}^2$. As the costs $-\left\|c_{0 \mid k}^{\star}\right\|_{\Pi}^2$ are non-positive, the sequence of $\{J_i\}_{i=0}^{\infty}$ is non-increasing. If $J_0$ is bounded then the sequence converges to $J_{\infty}$. By summing the left-hand-side of the inequality from $k =0$ to $\infty$, we have $\infty>J_0-J_{\infty} \geq \sum_{k=0}^{\infty}\left\|c_{0 \mid k}^{\star}\right\|_{\Pi}^2 \geq 0$. This implies $\lim _{k \rightarrow \infty} c_{0 \mid k}^{\star \top} \Pi c_{0 \mid k}^{\star}=0$ as $\Pi$ is positive definite. This proves the first statement.

Next, as $s_{k+1} = \Phi s_k+B c_{0 \mid k}^{\star}+d_k+ Dw_{k}$, we have 
$$\begin{aligned} & \lim _{k \rightarrow \infty} s_k \\ & =\lim _{k \rightarrow \infty}\left[\Phi^k s_0+\sum_{j=1}^k \Phi^{j-1} B c_{0 \mid k-j}^{\star}+\sum_{j=1}^k \Phi^{j-1}\left(d_{k-j}+D w_{k-j}\right)\right] \\ & =\lim _{k \rightarrow \infty}\left[\sum_{j=1}^k \Phi^{j-1} d_{k-j}\right]+\lim _{k \rightarrow \infty}\left[\sum_{j=1}^k D w_{k-j}\right]\end{aligned}
$$
Then, from the definition of the reachable set, with the probability  $\min\{1-\beta, 1- \delta_h\}$,  we have 
 $$
\begin{aligned}
&\lim _{k \rightarrow \infty} s_{k} \in \lim _{k \rightarrow \infty}\left\{0\right\} \oplus R_{\infty}^{e} \oplus R_{\infty}^{\varepsilon}, \\
&\lim _{k \rightarrow \infty} u_{k} \in \lim _{k \rightarrow \infty}\left\{0\right\} \oplus K R_{\infty}^{e} \oplus K R_{\infty}^{\varepsilon}. 
\end{aligned}
$$

Also, given the design of observables, we have $x_k = C s_k$; hence,  $\lim _{k \rightarrow \infty} x_{k} \in \lim _{k \rightarrow \infty}\left\{0\right\} \oplus CR_{\infty}^{e} \oplus CR_{\infty}^{\varepsilon}$.


\end{proof}
\section{Numerical Example}
\label{sec:num_exam}
In this section, we test and illustrate the Koopman MPC framework we have proposed in \eqref{eq:opt_prob} on a disturbed mass-spring system and investigate the impact of data samples and ball radii on the constraint satisfaction applying the recursively feasible approach. For all the cases, the optimization problems are formulated with CVXPY \cite{diamond2016cvxpy}, and solved by Gurobi \cite{gurobi2021gurobi}.
All the experiments have been run on a laptop MacBook Pro 2019 with a 2.8 GHz Quad-Core Intel Core i7 CPU.

\subsection{Experiment 1}
This experiment discretizes the continuous time model with $\dot{x} = 
\begin{bmatrix} 
    x_1 \\
    -x_1^{3} - 1.5x_2 + 0.5u
\end{bmatrix}$
and the sampling period $\Delta T = 0.1$ $\,s$. The prediction horizon is set to $N=5$ and the additive disturbance of the discrete-time system is $w = \begin{bmatrix} w_1, w_2 \end{bmatrix}^{\top}$, where $w_1$ and $w_2$ comply with a uniform distribution between $-0.001$ and $0.001$ and between $-0.1$ and $0.1$, respectively.  We generate $50000$ training data points for the data set $\mathcal{D}$ by simulating 500 trajectories over 100 sampling periods. Each element of the initial states follows the distribution $\mathcal{N}(0,1.5)$ and each control action is randomly generated following the distribution $\mathcal{N}(-7,5)$, similar to \cite{korda2018linear}. The lifted system is constructed with the data-set $\mathcal{D}$ and observables  $\begin{bmatrix}
    1 & x_1 & x_2 & x_1^2 & x_2^2 & x_1^3 & x_2^3
\end{bmatrix}^{\top}.$ 

Another data-set $D_{w}$ containing $330$ data points is generated by taking the difference between two consecutive states with zero-valued input close to the equilibrium point to estimate the additive disturbances. The disturbance support set is determined from $\mathcal{D}_w$ and the corresponding chance constraints tightening is acquired by solving the corresponding distributionally robust optimization with the ball radius 1e-4.

\update{The ball radius is determined via cross-validation as described below.   Specifically, we test the sensitivity of backoff $\eta$ - the solution of \eqref{eq:CVaR_DRO_reform} - with respect to varying numbers of i.i.d. samples and ball radii. In Table \ref{table:uniform_CVAR_DRO}, \updateNew{we present the sensitivity $\eta$ for varying numbers of samples and ball radii. The values in Table \ref{table:uniform_CVAR_DRO} are obtained by solving Eq. \eqref{eq:CVaR_DRO_reform} with different sample sizes and ball radii, representing the tightening estimated with finite samples of disturbance.}
We observe that as the radius increases, $\eta$ increases towards a robust tightening. Furthermore, with the same radius, $\eta$ approaches to the true CVaR value (0.09) as the number of samples increases. Based on these observations, we select a ball radius of $10^{-4}$ for the distributionally robust optimization, as it yields a backoff that is slightly conservative but remains close to the true value when enough samples are available.
\begin{table}
\centering
\begin{tabular}{|c|c|c|c|c|c|c|c|}
\hline
 $\varepsilon$ & $10^{-1}$ & $10^{-2}$   &  $10^{-3}$ & $10^{-4}$ & $10^{-5}$  & $10^{-6}$  & $10^{-7}$   \\
 \hline
 5 samples & 0.1 &  0.1 & 0.01711 & 0.00811 & 0.00721 & 0.00712 &  0.00711  \\
 \hline
 10 samples & 0.1 & 0.1 & 0.09747 & 0.08847 & 0.08757 & 0.08748  & 0.08747    \\
 \hline
 100 samples  & 0.1 & 0.1 & 0.09747 & 0.08847 & 0.08757 & 0.08748 & 0.08747 \\
 \hline
 500 samples  & 0.1  & 0.1 &  0.09986 & 0.09086 & 0.08996 & 0.08987 & 0.08986\\
 \hline
 1000 samples  & 0.1 & 0.1 & 0.09998 & 0.09098 & 0.09008 & 0.08999 & 0.08998 \\
 \hline
\end{tabular}
\caption{\label{table:uniform_CVAR_DRO} Sensitivity test of $\eta$ in terms of different numbers of samples and ball radii. The true CVaR value of $\eta$ corresponding to the uniform distribution $U(-0.1,0.1)$  is $0.09$ \cite{rockafellar2002conditional}.}
\end{table}
}

A feedback gain $K$ and penalty matrix $\Pi$ are determined via the discrete-time LQR with $Q = \text{diag}(1,100,100,\allowbreak 0.1,0.1,0.1,0.1)$ and R = 0.1.  We require $\mathbb{P}\{x_2 \le 0.6 \} \ge 0.9$. State constraints are then tightened by the procedure we introduced in Section IV.A.; by solving linear optimization problems, i.e. \cite[Theorem 8.1]{kouvaritakis2016model}. With the designed objective function and cost functions, we solve the optimization problem \eqref{eq:opt_prob} repeatedly.

From Fig \ref{fig:simulation} we can see that the state chance constraint is satisfied and the closed-loop state could be robustly regulated to the equilibrium in expectation. Additionally, the computational effort for the prediction horizon $5$ is 19.9 ms ± 1.8 ms per loop.

\update{Furthermore, we compare our method with the method introduced in \cite{paulson2019mixed} with the perfect knowledge of the model. For the chance constraints satisfaction and recursive feasibility, we design the tightening with the linearized model around the origin and the feedback gain is designed with $Q =\text{diag}(100,100) $ and $R = 0.1$ for the linearized system. From Fig. \ref{fig:simulation_compare_SMPC} we can see that the resulting closed-loop trajectories are our method is slightly more conservative than the ones with the perfect model information. The conservativeness is  caused due to the estimation of the model.}

\update{It is important to note that the quality of closed-loop performance depends on the choice of observables. Poorly chosen observables can lead to an overestimation of the uncertain support set, resulting in overly conservative closed-loop behavior. For further discussion on the role of observables, we refer the interested reader to \cite{korda2020optimal}.}

\subsection{Experiment 2}
 In this subsection, we compared our method with the robust method proposed in \cite{zhang2022robust}. We consider the beta-distributed additive disturbances 
$$
\begin{bmatrix} 0.002(\operatorname{Beta}(100, 100)-0.5)\\ (\operatorname{Beta}(100, 100)-0.5)\end{bmatrix},
$$
which are bounded yet more concentrated around their mean value. \update{Here, Beta distribution $\text{Beta}(\alpha,\beta)$ is understood as a family of continuous probability distributions with their probability density function defined via $\frac{x^{\alpha-1}(1-x)^{\beta-1}}{B(\alpha, \beta)}, \quad 0<x<1$ with $B(\alpha, \beta)=\int_0^1 t^{\alpha-1}(1-t)^{\beta-1} d t$.} Applying a robust control method, such as the one proposed by \cite{zhang2022robust}, to nonlinear systems with these additive disturbances would leads to closed-loop trajectories that are further from the constraints, thereby resulting in a lower convergence rage. This occurs because the worst-case, low-probability disturbance realization is treated with the same weight as the more frequent one around the mean value. This observation also motivates us to partition the errors into deterministic modeling errors and stochastic errors, which arise from the additive disturbances.
With the same data-collection and control process for both our approach and the robust approach proposed in \cite{zhang2022robust}, we can read from Fig. \ref{fig:compare} that our method results in a less conservative closed-loop performance.

From Fig. \ref{fig:compare} we can see that both methods can guarantee the satisfaction of the expected and 90-percent chance constraints. However, the trajectories resulting from \cite{zhang2022robust}  are more conservative as the constraints are tightened with respect to the worst-case uncertainties. Since the additive disturbances in this comparison are concentrated around zero, tightening the constraints with respect to the worst-case scenario would result in more conservativeness. However, in our method, we only consider state chance constraints, which are satisfied with the help of a data-driven distributionally robust optimization. Also, since we want to guarantee recursive feasibility, the worst-case uncertainties will be considered only after the first step of prediction. This difference resembles the difference between robust MPC and stochastic MPC.

\begin{figure}[thpb]
    \centering
    \includegraphics[width=0.8\textwidth]{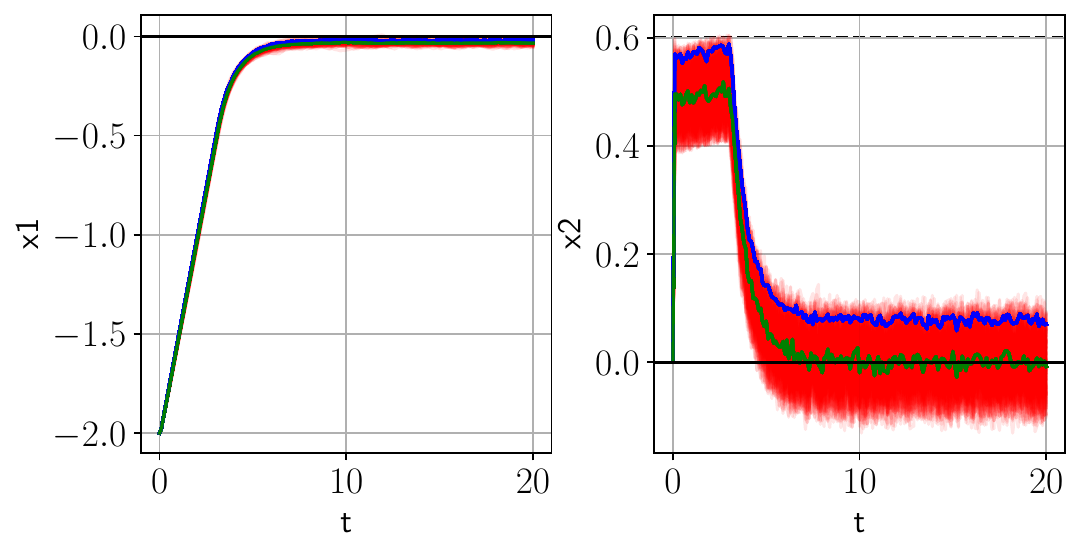}
    \caption{Red: 500 trajectories of different realizations. Green: Expected trajectory. Blue: 90-quantile trajectory of 500 realizations.}
    \label{fig:simulation}
\end{figure}

\begin{figure}[thpb]
    \centering
    \includegraphics[width=0.8\textwidth]{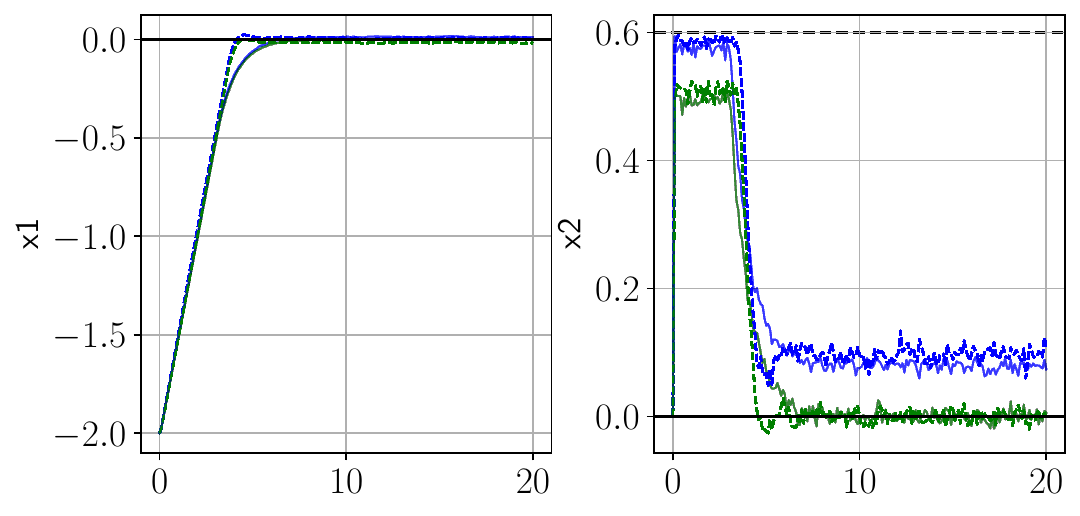}
    \caption{
    Solid light green: Expected trajectory using our method. Dashed dark green: Expected trajectory using method in \cite{paulson2019mixed}. Solid light Blue: 90-quantile trajectory of 500 realizations using our method. Dashed dark blue: 90-quantile trajectory of 500 realizations using method in \cite{paulson2019mixed}.}
    \label{fig:simulation_compare_SMPC}
\end{figure}

\begin{figure}[h]
    \centering
\includegraphics[width=0.8\textwidth]{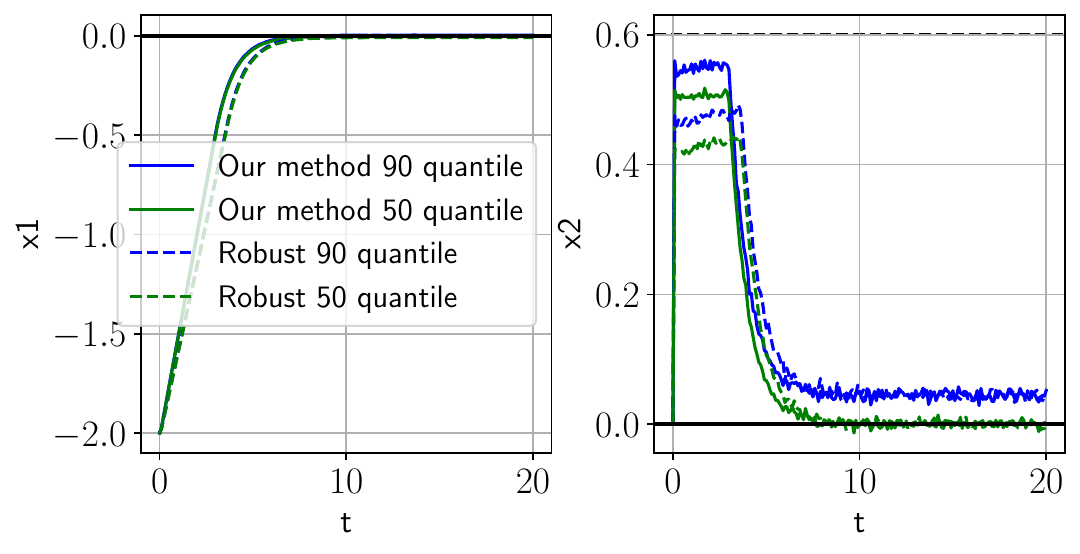}
    \caption{Comparison between the method proposed in our paper and in \cite{zhang2022robust}. Green indicates expected trajectory of 100 realizations with the same initial point and blue indicates the 90-quantile trajectory of 500 realizations.}
    \label{fig:compare}
\end{figure}

\section{CONCLUSION}
In this paper, we propose a novel data-driven Koopman MPC scheme for nonlinear systems. This scheme learns nonlinear systems in a lifted space with the help of the Koopman operator. A mixed-tubed-based MPC is proposed to control the lifted system and address the approximation error and stochastic disturbances such that the corresponding original system is regulated without violating the prescribed constraints. A constraint tightening is derived by solving a DRO problem with a Wasserstein ambiguity for chance constraints and a linear program for worse-case modeling errors. \update{This control scheme can be efficiently solved through a quadratic program with  guarantees on recursive feasiblility and closed-loop robustness in probability.} The effectiveness of the scheme is illustrated through a case study on a mass-spring system, and a less conservative closed-loop performance can be observed compared to a robust approach.

\bibliographystyle{unsrt}
\bibliography{manuscript}  
\end{document}